\documentclass{IEEEtran}

\usepackage{journal_masa-defs}

\title{Rate and Power Allocation for Discrete-Rate Link Adaptation}

\author{Anders~Gjendemsj{\o}, 
    Geir~E.~{\O}ien, 
    Henrik~Holm, 
    Mohamed-Slim~Alouini,\\ 
    David~Gesbert, 
    Kjell~J.~Hole, 
    and P{\aa}l Orten
    \thanks{This paper was presented in part 
    at the IEEE Global Telecommunications Conference, St. Louis, MO, USA, November 2005, and at the IEEE International Conference on Communications, Istanbul, Turkey, June 2006.}
    \thanks{A. Gjendemsj{\o} (contact author, email: gjendems@iet.ntnu.no) and G.~E.~{\O}ien
    are with the Dept.\ of Electronics and Telecommunications, Norwegian
    University of Science and Technology (NTNU), NO-7491 Trondheim, Norway.}
  \thanks{H. Holm was with the Dept. of Electronics and Telecommunications, NTNU, Trondheim, Norway. He is now with Honeywell Laboratories, Minneapolis, MN 55418, USA.}
  \thanks{M.-S.~Alouini is with the Dept. of Electrical and Computer Eng., TAMU-Q, Education City, Doha, Qatar.} 				
  \thanks{D.~Gesbert is with the Institut Eur{\'e}com, F-06904 Sophia-Antipolis, France.}
  \thanks{K. J.~Hole is with the Department of Informatics, University of Bergen, NO-5020 Bergen, Norway.}
  \thanks{P. Orten is with Thrane \& Thrane, NO-1375 Billingstad, Norway and with the University Graduate Center, Oslo, Norway.}}

\begin{document}
\maketitle
\begin{abstract}
Link adaptation, in particular adaptive coded modulation (ACM), is a promising tool for bandwidth-efficient transmission in a fading environment.  The main motivation behind employing ACM schemes is to improve the spectral efficiency of wireless communication systems. In this paper, using a finite number of capacity achieving component codes, we propose new transmission schemes employing constant power transmission, as well as discrete and continuous power adaptation, for slowly varying flat-fading channels.
 
We show that the proposed transmission schemes can achieve throughputs close to the Shannon limits of flat-fading channels using only a small number of codes. Specifically, using a fully discrete scheme with just four codes, each associated with four power levels,	we achieve a spectral efficiency within $1\,\dB$ of the continuous-rate continuous-power Shannon capacity. Furthermore, when restricted to a fixed number of codes, the introduction of power adaptation has significant gains with respect to ASE and probability of no transmission compared to a constant power scheme.
\end{abstract}



\section{Introduction}
\label{sec: introduction}
In wireless communications bandwidth is a scarce resource. By employing link adaptation, in particular adaptive coded modulation (ACM), we can achieve bandwidth-efficient transmission schemes. Today, adaptive schemes are already being implemented in wireless systems such as Digital Video Broadcasting - Satellite Version 2 (DVB-S2)~\cite{DVBS2}.
A generic ACM system
\cite{goldsmith/chua97,
  goldsmith/chua98,
  hole/holm/oien00,
  hole/oien01,
  chung/goldsmith01,
  hanzo/wong/yee02,
  holm02,
  catreux/erceg/gesbert/heath02,
  torrance/hanzo96}
is illustrated in Fig.~\ref{fig: system model}.
Such a system adapts to the channel variations by utilizing a set of component channel codes and modulation constellations with different spectral efficiencies (SEs).%

We consider a wireless channel with additive white Gaussian noise (AWGN) and
fading. Under the assumption of slow, frequency-flat fading, a block-fading model can be used to approximate the wireless fading channel by an AWGN channel within the length of a codeword~\cite{mceliece/stark84, ozarow/shamai/wyner94}.
Hence, the system may use codes which typically guarantee a certain spectral efficiency within a range of signal-to-noise ratios
(SNRs) on an AWGN channel.  At specific time instants, a prediction of the instantaneous
SNR is utilized to decide the highest-SE code that can be used. 
The system thus compensates for periods with low SNR by transmitting at a low SE, while transmitting at a high SE when the SNR is favorable.  In this way, a significant overall gain in \emph{average spectral efficiency}
(ASE)---measured in information bits/s/Hz--- 
can be achieved compared to fixed rate transmission systems.

\begin{figure}[tbp]
  \centering
    \includegraphics[width=3.2in]{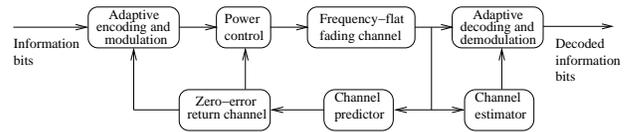}
    \caption{Adaptive coded modulation system.}
  \label{fig: system model}
\end{figure}


In the current literature we can identify two main approaches to the design of adaptive systems with a finite number of transmission rates~\cite{goldsmith/chua97, kose/goeckel00, paris/aguayo-torres/entrambasaguas01, byoungjo/hanzo03, lin/yates/spasojevic03, kim/skoglund05,kim/skoglund07}. One key point is the starting point for the design. In~\cite{lin/yates/spasojevic03, kim/skoglund05, kim/skoglund07} the problem can be stated as follows: Given that the system quantizes any channel state to one of $L$ levels, what is the maximum spectral efficiency that can be obtained using discrete rate signalling? On the other hand, in~\cite{goldsmith/chua97, kose/goeckel00, paris/aguayo-torres/entrambasaguas01, byoungjo/hanzo03} the question is: Given that the system can utilize $N$ transmission rates, what is the maximum spectral efficiency?
Another key difference is that
in~\cite{goldsmith/chua97, kose/goeckel00, paris/aguayo-torres/entrambasaguas01, byoungjo/hanzo03} the system is designed to maximize the average spectral efficiency according to a \textit{zero information outage} principle, such that at poor channel conditions transmission is disabled, and data buffered. 
However, in~\cite{lin/yates/spasojevic03, kim/skoglund05,kim/skoglund07}, data are allowed to be transmitted at all time instants, and an \textit{information} outage occurs when the mutual information offered by the channel is lower than the transmitted rate.
While seemingly similar, these approaches actually leads to different designs as will be demonstrated. Though allowing for a non-zero outage can offer more flexibility in the design, it also comes with the drawbacks of losing data and spilling system resources (e.g., power).
Furthermore, in~\cite{lin/yates/spasojevic03, kim/skoglund05,kim/skoglund07} the important issues of how often data are lost due to an information outage, and how to deal with it are not discussed, e.g., many applications would then require the communication system to be equipped with a retransmission capability.
These differences render a fair comparison between the approaches difficult; however we provide a numerical example later to illustrate the key points above.

In~\cite{lin/yates/spasojevic03, kim/skoglund05,kim/skoglund07} adaptive transmission with a finite number of capacity-achieving codes, and a single power level per code are considered. However, from previous work by Chung and Goldsmith~\cite{chung/goldsmith01}
we know that the spectral efficiency of such a restricted adaptive system increases if more degrees of freedom are allowed. In particular, for a finite number of transmission rates, power control is expected to have a significant positive impact on the system performance, and hence in this paper we propose and analyze more flexible power control schemes, for which the single power level per code scheme of~\cite{lin/yates/spasojevic03, kim/skoglund05,kim/skoglund07} can be seen as a special case.

\begin{figure}[tbp]
  \centering
    \includegraphics[width=3.2in]{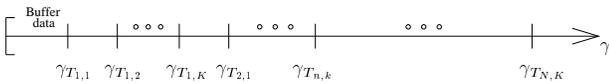}
  \caption{The pre-adaptation SNR range is partitioned into regions where $\swl{n,k}$~are the switching thresholds.}
  \label{fig: Fading axis bins}
\end{figure}

In this paper we focus on data communications which, as emphasized in~\cite{goldsmith05}, cannot ``tolerate any loss''. For such applications it thus seems more reasonable to follow the zero information outage design philosophy of~\cite{goldsmith/chua97, kose/goeckel00,paris/aguayo-torres/entrambasaguas01, byoungjo/hanzo03}. This choice is also supported by the work done in the design of adaptive coding and modulation for real-life systems, e.g., in DVB-S2~\cite{DVBS2}.  Based on this philosophy we derive transmission schemes that are optimal with regard to maximal ASE for a given fading distribution. By the assuming codes to be operating at AWGN channel capacity,
we formulate constrained ASE maximization problems and proceed to find the optimal switching thresholds and power
control schemes as their solutions. Considering both constant power transmission
as well as discrete and continuous power adaptation, we show that the introduction of power adaptation
provides substantial average spectral efficiency and probability  of no transmission gains when the number
of rates is finite.  Specifically, spectral efficiencies within $1\,\dB$ of the continuous-rate continuous-power Shannon capacity are obtained using a completely discrete transmission scheme with only four codes and four power levels per code.

The remainder of the present paper is organized as follows. We introduce the wireless model under investigation and describe the problem under study in Section~\ref{sec: system-model-problem}.
Optimal transmission schemes for link adaptation
are derived and analyzed in Section~\ref{sec: optimal-design-masa}. 
Numerical examples and plots are presented in Section~\ref{sec: numerical-exampl-results}.
Finally, conclusions and discussions are given in Section~\ref{sec: conclusions}.

\section{System Model and Problem Formulation}
\label{sec: system-model-problem}


\subsection{System Model}
\label{sec:system-model}

We consider the single-link wireless system
depicted in Fig.~\ref{fig: system model}. The discrete-time channel
is a stationary fading channel with time-varying
gain. The fading is assumed to be slowly varying and frequency-flat.
Assuming, as in~\cite{goldsmith/chua97, goldsmith/varaiya97}, that the transmitter receives perfect channel predictions
we can adapt the transmit power instantaneously at time $i$ according to a power adaptation scheme $S(\cdot)$.
Then, denote the instantaneous \emph{pre-adaptation} received signal-to-noise ratio (SNR) by
$\g[i]$, and the average pre-adaptation received SNR by $\gb$. These are the SNRs that would be experienced using signal
constellations of average power $\Sb$ without power control~\cite{hole/holm/oien00}.
Adapting the transmit power based on the channel state $\g[i]$, the received SNR after power control, termed \emph{post-adaptation} SNR, at time $i$ is then given by
$\g[i]S(\g[i])/\Sb$. By virtue of the stationarity assumption, the distribution of
$\g[i]$ is independent of $i$, and is denoted by $f_{\g}(\g)$.
To simplify the notation we omit the time reference $i$ from now on.

Following~\cite{gjendemsjo/oien/orten06,goldsmith/chua97},
we partition the range of $\g$ into $NK+1$ pre-adaptation SNR regions,
which are defined by the switching thresholds
$\lbrace\swl{n,k}\rbrace_{n,k=1,1}^{N,K}$, as illustrated in
Fig.~\ref{fig: Fading axis bins}.
Code $n$, with spectral efficiency $R_n$, is selected whenever $\g$ is in the interval $[\swl{n,1},\swl{n+1,1})$, $n=1,\cdots,N$.
Within this interval the transmission rate is constant, but the system can adapt the transmitted power, subject to an average power constraint of $\Sb$, to the channel conditions in order to maximize the average spectral efficiency.
If the pre-adaptation SNR is below $\swl{1,1}$, data are buffered. For convenience, we let $\swl{0,1}=0$ and $\swl{N+1,1}=\infty$.


\subsection{Problem Formulation}
The capacity of an AWGN channel is well known to be $C(\g)=\log_2\bigl(1+\frac{S(\g)}{\Sb}\g\bigr)$ information bits/s/Hz, where $\frac{S(\g)}{\Sb}\g$ is the received SNR\@.  This means that
there exist codes that can transmit with arbitrarily small error rate at all
spectral efficiencies up to $C(\g)$ bits/s/Hz, provided that the received SNR is (at least)~$\frac{S(\g)}{\Sb}\g$.
Our goal is now to find an optimal set of capacity-achieving transmission rates, switching levels, and power adaptation schemes in order to maximize the average spectral efficiency for a given fading distribution.

Using the results of~\cite{lin/yates/spasojevic03}, an information outage can only occur for a set of channel states within the first interval, which in our setup corresponds to that data should only be buffered for channel states in the first interval.  Whereas in the other SNR regions, the assigned rate supports the worst channel state of that region. The average spectral efficiency of the system (in information bit per channel use) can then be written as
\begin{equation}
\strek{R} = \sum_{n=1}^N R_n P_n,
\end{equation}
where $P_n=\int_{\swl{n,1}}^{\swl{n+1,1}}f_{\g}(\g) \, \dif \g$, which is the probability that code $n$ is used.

\section{Optimal Design for Maximum Average Spectral Efficiency}
\label{sec: optimal-design-masa}
Based on the above setup, we now proceed to design spectral efficiency maximizing schemes.
Recall that the pre-adaptation SNR range is divided into regions lower bounded
by $\swl{n,1}$, for $n=0,1,\cdots,N$.
Thus, we let $R_n = C_n$, where $C_n=\log_2\bigl(1+\frac{S(\swl{n,1})}{\Sb}\swl{n,1}\bigr)$
is shown below to be the highest spectral efficiency that can be supported within the range $[\swl{n,1},\swl{n+1,1})$ for
$1\leq n \leq N$, after transmit power adaptation. Note that the fading is nonergodic within each codeword, so that the results
of~\cite[Section IV]{caire/shamai99} do not apply.

An upper bound on the ASE of the ACM scheme---for a given set of
codes/switching levels---is therefore
the \emph{maximum ASE for ACM} (MASA), defined as
\begin{equation}
\begin{split}
  \MASA &= \sum_{n=1}^{N} C_n P_n\\
            &= \sum_{n=1}^N \log_2 \bigl( 1 + \frac{S(\swl{n,1})}{\Sb}\swl{n,1} \bigr )
         \int_{\swl{n,1}}^{\swl{n+1,1}} \f{\g}\,\dif \g,
\end{split}
  \label{eq: tot-ase}
\end{equation}
subject to the average power constraint,
\begin{equation}
\label{eq: Power constraint for MASA - Start} \sum_{n=0}^{N}
\int_{\swl{n,1}}^{\swl{n+1,1}}S(\g)f_{\g}(\g)\, \dif\g
\leq \Sb,
\end{equation}
where $\Sb$ denotes the average transmit power. \eqref{eq: tot-ase} is basically a
discrete-sum approximation of the
integral expressing the Shannon capacity in~\cite[Eq. (4)]{goldsmith/varaiya97}.
If arbitrarily long codewords can be used, the bound can
be approached from below with arbitrary precision for an
arbitrarily low error rate\@. Using $N$ distinct codes we analyze the MASA for constant,
discrete, and continuous transmit power adaptation schemes, deriving the optimal
rate and power adaptation for maximizing the average spectral efficiency.
We shall assume that the fading is so slow that capacity-achieving
codes for AWGN channels can be employed, giving tight bounds
on the MASA~\cite{dolinar/divsalar/pollara98-I,dolinar/divsalar/pollara98-II}.





\subsection{Continuous-Power Transmission Scheme}
\label{ssec:Optimal power adaptation}
In an ideal adaptive power control scheme, the transmitted power can be varied to entirely track the channel variations. Then, for the $N$ regions where we transmit, we show that the optimal continuous power adaptation scheme is \emph{piecewise channel inversion\footnote{The results of this section were in part presented in~\cite{gjendemsjo/oien/holm05}. Similar results on continuous power adaptation, when allowing for information outage (loss of data), were also later independently reported in~\cite{lin/yates/spasojevic06}.}} to keep the
received SNR constant within each region, much like the bit error rate is kept constant in optimal adaptation for constellation
restrictions in \cite{goldsmith/chua97}. For each rate region we use a capacity-achieving code
which ensures an arbitrarily low probability of error for any AWGN channel with a received SNR greater than or
equal to $\frac{S(\swl{n,1})}{\Sb}\swl{n,1}\triangleq\kappa_n$.
The optimality of this strategy is formally proven below.

\begin{lemma}
\label{lemma: optimal continuous power}
For the $N+1$ SNR regions the optimal
continuous power control scheme is of the form

\begin{equation}
\label{eq:New power adaptation scheme}
\frac{S(\g)}{\Sb} =
\begin{cases}
\frac{\kappa_n}{\g},
& \text{if } \swl{n,1}\leq \g <\swl{n+1,1},
\, 1\leq n \leq N,\\
0, & \text{if } \g < \swl{1,1},
\end{cases}
\end{equation}
where $\lbrace\kappa_n,\swl{n,1}\rbrace_{n=1}^N$ are parameters to be
optimized.
\end{lemma}
\begin{proof}
Assume for the purpose of contradiction that the power scheme given
in~\eqref{eq:New power adaptation scheme} is not optimal, i.e., it
uses too much power for a given rate. Then, by assumption, there
exists at least one point in the set
\begin{equation}
\bigcup_{n=1}^N\lbrace \g: \swl{n,1} \leq \g
<\swl{n+1,1}\rbrace
\end{equation}
where it is possible to use less power; denote this point by
$\g'$. Fix any $\epsilon>0$ and let
$\frac{S(\g')}{\Sb}=\frac{\kappa_n}{\g'}-\epsilon$.
This yields a received SNR of
$\kappa_n-\epsilon \g'<\kappa_n$, but is less than the
minimum required SNR for a rate of $\text{log}_2(1+\kappa_n)$. Hence,
it does not exist any point where the proposed power scheme can
be improved, and the assumption is contradicted.
\end{proof}

Using~\eqref{eq:New power adaptation scheme} the received SNR, after power adaptation, is then
for $n=1,2,\cdots,N$ given as:
\begin{equation}
\label{eq:SNR for new scheme}
\frac{S(\g)}{\Sb}\g =
\begin{cases}
\kappa_n,
& \text{ if } \swl{n,1}\leq \g <\swl{n+1,1},\\
0, & \text{ if } \g < \swl{1,1},
\end{cases}
\end{equation}
i.e., we have a constant received SNR of $\kappa_n$ within each region, supporting a
maximum spectral efficiency of $\log_2(1+\frac{S(\swl{n,1})}{\Sb}\swl{n,1})=\log_2(1+\kappa_n)$.

Introducing the power adaptation scheme~\eqref{eq:New power adaptation scheme} in~\eqref{eq: tot-ase},~\eqref{eq: Power constraint for MASA - Start}, and changing the average power inequality to an equality for maximization, we arrive at a scheme we denote $\MASANI$, posing the
following optimization problem\footnote{Strictly speaking, we should add the constraints $0\leq \swl{1,1}\leq \dots \leq \swl{N,1}$, and $\kappa_n\geq 0,\forall n$. However, we instead verify that the solutions we find satisfy these constraints.}:

\begin{subequations}\label{eq: continuous optimization problem}
\begin{align}
&\text{maximize}
\quad \MASANI = \sum_{n=1}^{N}\text{log}_2(1+\kappa_n)w_n\\
&\quad \quad
\text{s.t}
\quad \,\quad \sum_{n=1}^{N}\kappa_n c_n = 1,
\label{eq: contin pow constr}
\end{align}
\end{subequations}
where we have introduced the notation $w_n=\int_{\swl{n,1}}^{\swl{n+1,1}} f_\g(\g)\,\dif \g$, $c_n =\int_{\swl{n,1}}^{\swl{n+1,1}}\frac{1}{\g}f_{\g}(\g)\,\dif \g$.
Note that for $N=1$, \eqref{eq: continuous optimization problem} reduces to the truncated channel inversion Shannon capacity scheme given in~\cite[Eq. 12]{goldsmith/varaiya97}.
Inspecting~\eqref{eq: continuous optimization problem}, we see that for a given set of $\{\swl{n,1}\}$, the problem is a standard convex optimization problem in $\{\kappa_n\}$ with a waterfilling solution given as~\cite{cover/thomas91}
\begin{equation}
\label{eq: solution for beta continuous}
\kappa_n = \frac{w_n}{\lambda c_n}-1,\, n=1,\cdots,N,
\end{equation}
where $\lambda$ is a Lagrange multiplier to satisfy the average power constraint, and can from~\eqref{eq: contin pow constr} be expressed as a function of the switching thresholds as
\begin{equation}
\lambda = \frac{1-\CDF{\swl{1,1}}}{1+\sum_{n=1}^N c_n},
\end{equation}
where $F_{\g}(\cdot)$ denotes the cumulative distribution function (cdf) of $\g$.  Finally, the optimal values of $\{\swl {n,1}\}$, are found by equating the gradient of MASA$_{N\times \infty}$ to zero,
\begin{equation}
\label{eq: MASA cont all inserted}
\nabla \text{MASA}_{N\times \infty} = \nabla \sum_{n=1}^N\log_2(\frac{w_n}{\lambda c_n}) = \vek{0}.
\end{equation}

The solution to~\eqref{eq: MASA cont all inserted} is found through numerical procedures, reducing the $2N$-dimensional constrained optimization problem to solving $N$ equations in $N$ unknowns. Numerical results for the resulting adaptive power policy and the corresponding spectral efficiencies are presented in
Section~\ref{sec: numerical-exampl-results}.


\subsection{Discrete-Power Transmission Scheme}
\label{ssec: Discrete-Power Transmission Scheme}
For practical scenarios the resolution of power control will be limited, e.g., for the Universal Mobile Telecommunications System (UMTS) power control step sizes on the order of $1\,\dB$ are proposed~\cite{3GphysTDD}.
We thus extend the MASA analysis by considering discrete power adaptation.
Specifically, we introduce the MASA$_{N\times K}$ scheme
where we allow for $K\geq 1$ power regions \emph{within} each of the $N$ rate regions.
For each rate region we again use a capacity-achieving code
for any AWGN channel with a received SNR greater than or
equal to $\frac{S(\swl{n,1})}{\Sb}\swl{n,1}=\kappa_n$.
The optimal discrete power adaptation is discretized piecewise channel inversion,
closely related to the discrete power scheme in~\cite{paris/aguayo-torres/entrambasaguas01}.

\begin{lemma}
\label{lemma: optimal discrete power}
The optimal discrete-power adaptation scheme is of the form
\begin{equation}
\label{eq: Power adaptation scheme - N x K}
\frac{S(\g)}{\Sb} = 
\begin{cases}
\frac{\kappa_{n}}{\swl{n,k}}, &
 \text{if } \swl{n,k}\leq \g <\swl{n,k+1},\\
  & 1\leq n \leq N, \text{ } 1\leq k \leq K\\
0, & \text{if } \g < \swl{1,1},
\end{cases}
\end{equation}
where $\swl{n,K+1} \triangleq \swl{n+1,1}$. $\lbrace \kappa_n \rbrace_{n=1}^N$ and $\lbrace \swl{n,k}
\rbrace_{n,k=1,1}^{N,K}$ are the parameters to be optimized.
\end{lemma}
\begin{proof}
To ensure reliable transmission in each rate region $1\leq n \leq
N$, we require $\frac{S(\g)}{\Sb} \g \geq \kappa_n$, assuming $\g \in
[\swl{n,1}, \swl{n+1,1})$. Thus, following the proof of
Lemma~\ref{lemma: optimal continuous power}, since the rate is
restricted to be constant in each region, it is obviously optimal
from a capacity maximization perspective to reduce the transmitted
power, when the channel conditions are more favorable. \eqref{eq:
Power adaptation scheme - N x K} is then obtained by reducing the
power in a stepwise manner ($K-1$ steps), and at each step obtaining a received SNR of $\kappa_n$, i.e., $\frac{S(\swl{n,k})}{\Sb}\swl{n,k}=\kappa_n$, thus using the least possible power, while still ensuring transmission with an arbitrarily low error rate.
\end{proof}

Compared to the continuous-power transmission scheme~\eqref{eq:New power adaptation scheme}, discrete-level power control~\eqref{eq: Power adaptation scheme - N x K} will be suboptimal.
As seen from the proof of Lemma~\ref{lemma: optimal discrete power}, this is due to fact that~\eqref{eq: Power adaptation scheme - N x K} is only optimal at $K$ points ($\swl{n,1},\cdots,\swl{n,K}$) within each pre-adaptation SNR region $n$, at all other points the transmitted power is greater than what is required for reliable transmission
at $\log_2(1+\kappa_n)$ bits/s/Hz. Clearly, increasing the number of power levels per code $K$ gives a better approximation to the continuous power control~\eqref{eq:New power adaptation scheme}, resulting in a higher
average spectral efficiency. However, as we will see from the numerical results in
Section~\ref{sec: numerical-exampl-results}, using only a few power levels per code will yield spectral efficiencies
close to the upper bound of continuous power adaptation.

Using~\eqref{eq: Power adaptation scheme - N x K}
in~\eqref{eq: tot-ase},~\eqref{eq: Power constraint for MASA -
Start},
we arrive at the following optimization problem: 

\noindent

\begin{subequations}
\label{eq: MASA N x K - power scheme inserted}
\begin{align}
\text{maximize}\quad \sum_{n=1}^N \log_2(1+\kappa_{n}) w_n\\
\text{s.t.}\qquad \sum_{n=1}^N \kappa_n d_{n} = 1,\label{eq: constraint n k}
\end{align}
\end{subequations}
where we have introduced $d_{n}= \sum_{k=1}^{K} \frac{1}{\swl{n,k}} \int_{\swl{n,k}}^{\swl{n,k+1}} f_\g(\g)\, \dif\g$. As in the case of continuous-power transmission, for fixed $\{\swl{n,k}\}$,~\eqref{eq: MASA N x K - power scheme inserted} is a standard convex optimization problem in $\lbrace \kappa_n\rbrace$, yielding optimal values according to waterfilling as
\begin{equation}
\label{eq: solution for beta discrete}
\kappa_n = \frac{w_n}{\lambda d_{n}}-1,\, n=1,\cdots,N,
\end{equation}
where again $\lambda$ is a Lagrange multiplier for the power constraint, and from~\eqref{eq: constraint n k} expressed as
\begin{equation}
\lambda = \frac{1-\CDF{\swl{1,1}}}{1+\sum_{n=1}^N d_{n}}
\end{equation}

Then, the optimal switching thresholds $\{\swl{n,k}\}_{n=1,k=1}^{N,K}$ are found by numerically solving the following equation
\begin{equation}
\label{eq: MASA disc all inserted}
\nabla \text{MASA}_{N\times K} = \nabla \sum_{n=1}^N\log_2(\frac{w_n}{\lambda d_{n}}) = \vek{0},
\end{equation}
again simplifying the original constrained optimization problem to solving a set of equations.



\subsection{Constant-Power Transmission Scheme}
\label{sec: constant-power-trans}
When a single transmission power is used for all codes, we adopt the
term \emph{constant-power transmission scheme}\footnote{Also termed \emph{on-off} power transmission, see e.g.,~\cite{chung/goldsmith01}.}~\cite{kose/goeckel00}.
The optimal constant power policy is then to save power when $\g<\swl{1,1}$,
i.e., when there is no transmission, while transmitting at a power level
satisfying~\eqref{eq: Power constraint for MASA - Start} with an equality for
$\g\geq \swl{1,1}$, i.e.,
\begin{equation}
\label{eq: Constant power adap scheme}
\frac{S(\gamma)}{\Sb} =
\begin{cases}
\frac{1}{1-F_{\g}(\swl{1,1})},
& \text{if } \swl{n,1} \leq \g <\swl{n+1,1},
\, 1\leq n \leq N,\\
0, & \text{if } \g < \swl{1,1}.
\end{cases}
\end{equation}
From~\eqref{eq: Constant power adap scheme}
we see that the post-adaptation SNR monotonically increases
within $[\swl{n,1}, \swl{n+1,1})$, for $1\leq n \leq N$. Hence,
$\log_2 \bigl( 1+ \frac{S(\swl{n,1})}{\Sb}\swl{n,1} \bigr)$ is the highest possible
spectral efficiency that can be supported over the whole of region $n$.
Introducing~\eqref{eq: Constant power adap scheme} in~\eqref{eq: tot-ase}
we obtain a new expression for the MASA, denoted by MASA$_N$:
\begin{equation}
\text{MASA}_N =
\sum_{n=1}^{N}\text{log}_2\Bigl(1+\frac{\swl{n,1}}{1-F_{\g}(\swl{1,1})}\Bigr)
\int_{\swl{n,1}}^{\swl{n+1,1}}f_{\g}(\g) \, \dif\g.
\label{eq:MASA for AWGN fading and constant power}
\end{equation}
In order to find the optimal set of switching levels
$\{\swl{n,1}\}_{n=1}^{N}$, we first calculate the gradient of the MASA$_N$---as
defined by~\eqref{eq:MASA for AWGN fading and constant power}---with respect to the switching
levels. The gradient is then set to zero,
and we attempt to solve the resulting set of equations with respect to
$\{\swl{n,1}\}_{n=1}^N$:
\begin{equation}
  \nabla\MASA_N =
  \begin{bmatrix}
    \frac{\partial\MASA_N}{\partial\swl{1,1}} \\
    \vdots \\
    \frac{\partial\MASA_N}{\partial\swl{N,1}}
  \end{bmatrix}
  = \vek{0}.
  \label{eq:MASA-gradient}
\end{equation}

For $n=2,\cdots,N$ the partial derivatives in~\eqref{eq:MASA-gradient} can be expressed as follows:
\begin{equation}
  \begin{split}
    &\frac{\partial\MASA_N}{\partial\swl{n,1}}
      =\log_2(e)\times   \Biggl(\frac{\int_{\swl{n,1}\f{\g}\,\dif\g}^{\swl{n+1,1}}}{1-\CDF{\swl{1,1}}+\swl{n,1}}\\ & -\ln\left(\frac{1-\CDF{\swl{1,1}}+\swl{n,1}}{1-\CDF{\swl{1,1}}+\swl{n-1,1}}\right) \f{\swl{n,1}}\Biggr),
  \end{split}
  \label{eq:derivative}
\end{equation}
where $\ln(\cdot)$ is the natural logarithm. 
   The integral in~\eqref{eq:derivative} is recognized as the
difference between the cdf of
$\g$, $\CDF{\cdot}$, evaluated at the two points $\swl{n+1,1}$
and $\swl{n,1}$.  Setting $\frac{\partial\MASA_N}{\partial\swl{n,1}}$ for $2\leq n \leq N$ equal to zero then yields a set
of $N-1$ equations, each with a similar form to the one shown here:
\begin{multline}
  \CDF{\swl{n+1,1}} - \CDF{\swl{n,1}}
    - (1-\CDF{\swl{1,1}}+\swl{n,1}) \times \\
      \ln\left(\frac{1-\CDF{\swl{1,1}}+\swl{n,1}}{1-\CDF{\swl{1,1}}+\swl{n-1,1}}\right) \f{\swl{n,1}}
       = 0.
\label{eq:cdf-derivative}
\end{multline}
Noting that $\swl{n+1,1}$ appears only in one place in this equation, it
is trivial to rearrange the $N-2$ first equations into a recursive set
of equations where $\swl{n+1,1}$ is written as a function of $\swl{n,1}$,
$\swl{n-1,1}$, and $\swl{1,1}$ for $n=2,\dots,N-1$:
\begin{equation}
\begin{split}
  \swl{n+1,1} = &\ICDF
    \bigl[
      \CDF{\swl{n,1}}
      +(1-\CDF{\swl{1,1}}+\swl{n,1})\times \\
      & \ln\left(\frac{1-\CDF{\swl{1,1}}+\swl{n,1}}{1-\CDF{\swl{1,1}}+\swl{n-1,1}}\right) \f{\swl{n,1}}
    \bigr]
\end{split}
\label{eq: cdf-recursive}
\end{equation}
where $\ICDF[\cdot]$ is the inverse cdf\@.

For $N\geq 3$, \eqref{eq: cdf-recursive} can be expanded in order to yield a set $\swl{3,1},\cdots,\swl{N,1}$ which is optimal for given $\swl{1,1}$ and $\swl{2,1}$. The MASA can then be expressed as a function of $\swl{1,1}$ and $\swl{2,1}$ only.
We have now used $N-2$ equations from the set
in~\eqref{eq:MASA-gradient}, and the two remaining equations could be used in order
to reduce the problem to one equation of one unknown.  However,
both because of the recursion and the complicated expression for $\frac{\partial\MASA_N}{\partial\swl{1,1}}$,
the resulting equation would become prohibitively involved.
The final optimization is done by numerical maximization of $\MASA_N(\swl{1,1},\swl{2,1})$, thus reducing the $N$-dimensional optimization problem to $2$ dimensions. After solving the reduced problem $\swl{3,1},\cdots,\swl{N,1}$ are found via~\eqref{eq: cdf-recursive}.


\section{Numerical Results}
\label{sec: numerical-exampl-results}
One important outcome of the research presented here is the
opportunity the results provide for
assessing the relative significance of the number of codes and power levels used. It
is in many ways desirable to use as few
codes and power levels as possible in link adaptation schemes, as this may help overcome several problems, e.g.,
relating to implementation complexity, adaptation with faulty
channel state information (CSI),
 Thus, if we
can come close to the maximum MASA (i.e., the channel capacity) with
small values of $N$ and $K$ by choosing our link adaptation schemes optimally,
this is potentially of great practical interest.

The constant and discrete schemes offer several advantages considering
implementation~\cite{digham/alouini04}. In these schemes the transmitter
adapts its power and rate from a limited set of values, thus the receiver only
need to feed back an indexed rate and power pair for each fading block. Obviously,
compared to the feedback of continuous channel state information, this results
in reduced requirements of the feedback channel bandwidth and transmitter design. Further, completely discrete schemes are more resilient towards errors in channel estimation and prediction.

Two performance merits will be taken
into account: The MASA, representing an approachable upper bound on the
throughput when the scheme is under the restriction of a certain number
of codes and power adaptation flexibility, and the probability
of no transmission, $\Pnotr$.


\subsection{Switching Levels and Power Adaptation Schemes}
\label{sec: switch-levels-pow-ad-schemes}
\begin{figure}
  \centering
  \includegraphics[width=3.2in]{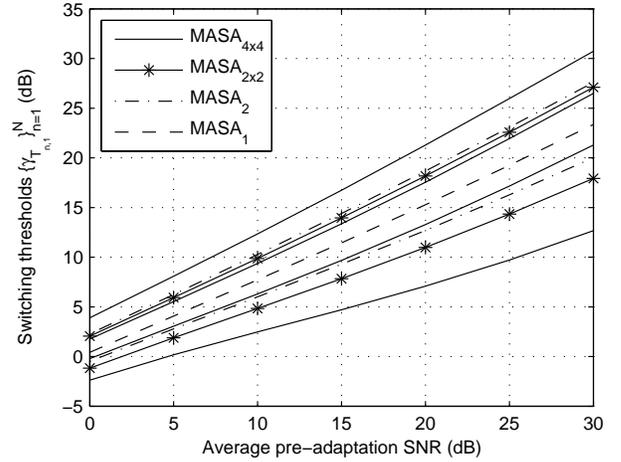}
  \caption{Switching thresholds $\{ \swl{n,1}\}_{n=1}^N$ as a function of average pre-adaptation SNR.
For each data series, the lowermost curve shows $\swl{1,1}$, while the uppermost shows $\swl{N,1}$.}
\label{fig: Thresholds_1_2_2x2_4x4}
  \label{fig: Switching Thresholds}
\end{figure}


Fig.~\ref{fig: Switching Thresholds} shows the set of optimal switching
levels $\{\swl{n,1}\}_{n=1}^{N}$ for
selected MASA schemes and for $0\,\dB<\gb<30\,\dB$.
(For the $\MASANKa{2}{2}$ and $\MASANKa{4}{4}$ schemes the internal
switching thresholds $\lbrace \swl{n,k} \rbrace_{n=1, k=2}^{N,K}$ are not shown in
Fig.~\ref{fig: Switching Thresholds} due to clarity reasons).
Table~\ref{tab: rate example} shows numerical values,
correct to the first decimal place, for designing optimal systems
with $N=4$ at $\gb=10\,\dB$.
Fig.~\ref{fig: Switching Thresholds} and Table~\ref{tab: rate example}
should be interpreted as follows:  With the mean pre-adaptation SNR $\gb$,
the number of codes $N$ and a power adaptation scheme in
mind, find the set of switching levels and the corresponding
maximal spectral efficiencies, given by
\begin{equation}
\label{eq: Spectral efficiencies} \text{SE}_n=
\begin{cases}\text{log}_2(1+\frac{\swl{n,1}}{1-F(\swl{1,1})})\quad \text{for MASA}_N,\\
\text{log}_2(1+\kappa_n)\quad \text{for MASA}_{N\times K}\quad \text{and}\quad \MASANI .
\end{cases}
\end{equation}
Then design optimal codes for these spectral efficiencies, for each $\gb$ of
interest.


\begin{table*}
\renewcommand{\arraystretch}{1.3}
\caption{Rate and {P}ower {A}daptation for {F}our {R}egions,
$\gb=10 \text{ }\mathrm{dB}$} \label{tab: rate example}
\centering
\begin{tabular}{l|c|c|c}
\hline & MASA$_4$ & MASA$_{4\times 4}$ & MASA$_{4\times \infty}$\\
\hline
$\swl{1,1},\cdots, \swl{4,1}\, (\dB)$ & $4.4, 7.3, 9.8, 12.4$ & $2.5, 6.3, 9.4, 12.3$
&$1.4, 5.5, 8.9, 12.3$\\
$\kappa_1,\cdots,\kappa_4$ & - & $2.4, 6.6, 13.9, 29.0$ & $2.0, 6.0, 13.8, 31.3$\\
SE$_1,\cdots,$SE$_4$ & $1.9, 2.7, 3.4, 4.2$ &$1.8, 2.9, 3.9, 4.9$ &$1.6, 2.8, 3.9, 5.0$\\
\hline
\end{tabular}
\end{table*}



Examples of optimized power adaptation schemes are shown in
Fig.~\ref{fig: Power_scheme_C_opra_4xinf_4x4_10_db}, illustrating
the piecewise channel inversion power adaptation schemes of the
$\MASANK$ and $\MASANI$ schemes. For $\g \leq 15\,\dB$ the
discrete-power scheme of $\MASANKa{4}{4}$ closely follows the
continuous power adaptation scheme of $\MASANIa{4}$. Fig.~\ref{fig:
Power_scheme_C_opra_4xinf_4x4_10_db} also depicts the optimal power
allocation (denoted $C_\text{OPRA}$) for continuous-rate
adaptation~\cite[Eq. 5]{goldsmith/varaiya97}. At $\gb=10\,\dB$, two
discrete-rate MASA schemes  allocate more power to codes with higher
spectral efficiency, following the water-filling nature of
$C_\text{OPRA}$. In the analysis of Section~\ref{sec:
optimal-design-masa} no stringent peak power constraint has been
imposed, and it is interesting to note the limited range of $S(\g)$
that still occurs for both $\MASANKa{4}{4}$ and $\MASANIa{4}$.
\begin{figure}
  \centering
  \includegraphics[width=3.2in]{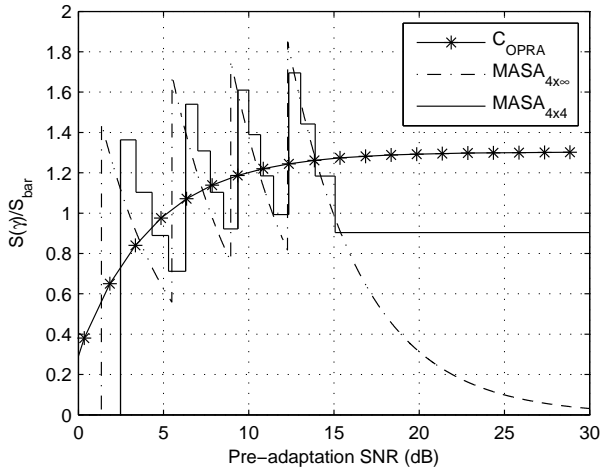}
  \caption{Power adaptation schemes for MASA$_{4\times \infty}$ and MASA$_{4\times 4}$
  as a function of pre-adaptation SNR, plotted for an average pre-adaptation SNR $\gb=10$ dB.
  Optimal power adaptation for continuous-rate adaptation $C_\text{OPRA}$ as reference.}
  \label{fig: Power_scheme_C_opra_4xinf_4x4_10_db}
\end{figure}

\subsection{Comparison of MASA Schemes}
\label{sec: MASA-comp}
Under the average power constraint of~\eqref{eq: Power constraint for MASA - Start} the
average spectral efficiencies corresponding to MASA$_N$, $\MASANK$, and $\MASANI$
are plotted in Figs.~\ref{fig: MASA ASE Plots} and~\ref{fig: MASA 4 codes}.
From Fig.~\ref{fig: MASA_N C_ORA} we see that the average spectral efficiency
increases with the number of codes, while Fig.~\ref{fig: MASA 4 codes} shows that the
ASE also increases with flexibility of power adaptation.


Fig.~\ref{fig: MASA_disc_k_VS_MASA_disc_k} compares four MASA
schemes with the product $N\times K=8$, showing that number of codes has a
slightly larger impact on the spectral efficiency than the number of power levels. However, we see that
the three schemes with $N\geq 2$ have almost similar performance, indicating that the number of rates and
power levels can be traded against each other, while still achieving
approximately the same ASE. From an implementation point of view
this is valuable as it gives more freedom to design the
system.

Finally, as mentioned in the introduction, there are at least two distinct design philosophies for link adaptation systems, depending on whether the number of \textit{regions} in the partition of the pre-adaptation range $\g$ or the number of \textit{rates} is the starting point of the design, and correspondingly on whether information \textit{outage} can be tolerated. Now, a direct comparison is not possible, but to highlight the differences between the two philosophies we provide a numerical example.
\begin{example}
\label{ex: outage prob from yates}
Consider designing a simple rate-adaptive system with two regions, where the goal is to maximize the expected rate using a single power level per region. Assuming the average pre-adaptation SNR on the channel to be $5\,\dB$ and following the setup of~\cite{lin/yates/spasojevic03, kim/skoglund05, kim/skoglund07}, 
we find the maximum \textit{average reliable throughput} (ART), defined as the ``average data rate assuming zero rate
when the channel is in outage''\cite{lin/yates/spasojevic03} that can be achieved to be $1.2444$ bits/s/Hz, and that the probability of information outage, or equivalent the probability that an arbitrary transmission will be corrupted, is $0.3098$. Thus, without retransmissions, the system is likely to be useless for many applications.

Now, turning to the MASA schemes discussed in this paper, using two regions, but only one constellation and power level, i.e., $\MASA_1$, we see from Fig.~\ref{fig: MASA_N C_ORA} that this scheme achieves a spectral efficiency of $1.2263$ bits/s/Hz at $\gb=5\,\dB$ without outage. This is only marginally less than the scheme from~\cite{lin/yates/spasojevic03, kim/skoglund05, kim/skoglund07} when using two constellations and allowing for a non-zero outage. 
\end{example}

\begin{figure}
  \centerline{
    \subfigure[Average spectral efficiency of MASA$_N$ for $N=1,2,4,8$ and $C_{\text{ORA}}$ for reference.]{
      \includegraphics[width=0.48\linewidth]{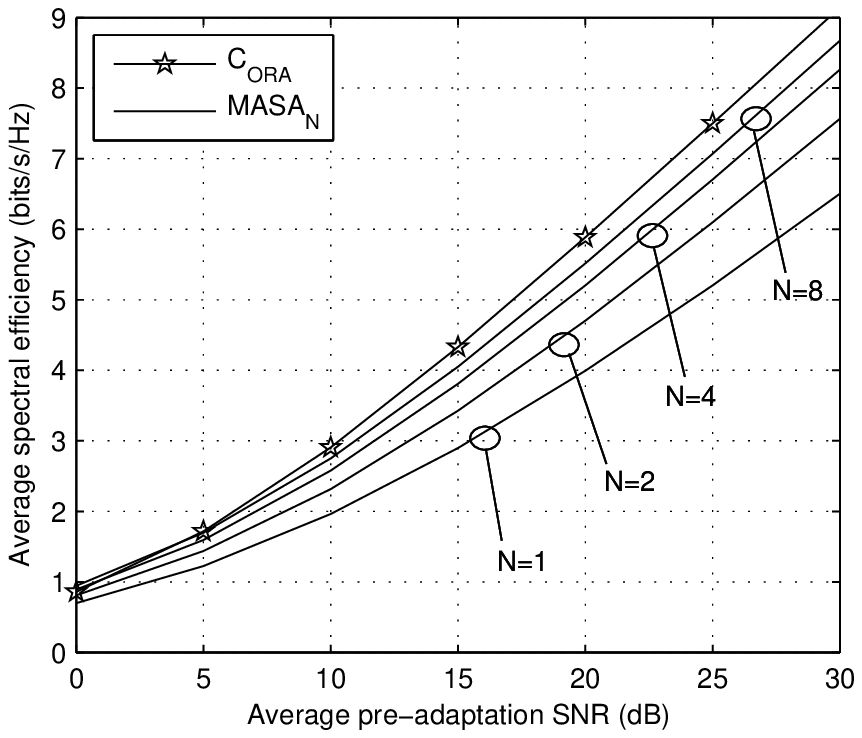}
     \label{fig: MASA_N C_ORA}}
    \hfil
    \subfigure[Average spectral efficiency of $\MASANK$ as a function of $\gb$,
                         for four MASA schemes with $N\times K=8$.]{
      \includegraphics[width=0.48\linewidth]{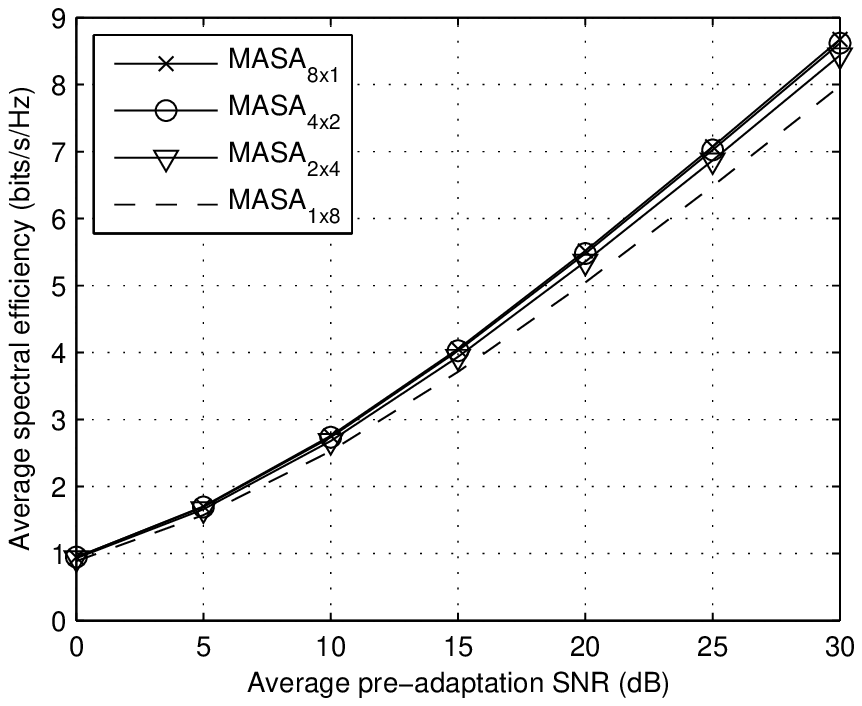}
     \label{fig: MASA_disc_k_VS_MASA_disc_k}}}
   \caption{Average spectral efficiency of different MASA schemes.}
  \label{fig: MASA ASE Plots}
\end{figure}

\begin{figure}
  \centering
   \includegraphics[width=3.2in]{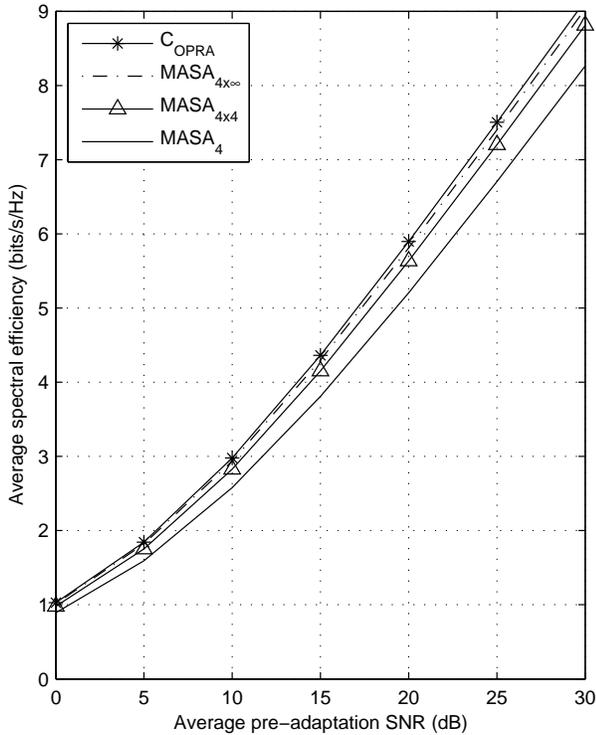}
   \caption{Average spectral efficiency for various MASA schemes with $N=4$ codes as a function of $\gb$.
   $C_{\text{OPRA}}$ as reference.}
  \label{fig: MASA 4 codes}
\end{figure}



\subsection{Comparison of MASA schemes with Shannon Capacities}
\label{sec: comp-with-capac}
Assume that the channel state information $\g$ is known to the
transmitter and the receiver. Then, given an average transmit
power constraint the channel
capacity of a Rayleigh fading channel with optimal \emph{continuous}
rate adaptation and constant transmit power, $C_{\text{ORA}}$, is
given in~\cite{goldsmith/varaiya97,alouini/goldsmith99} as

\begin{equation}
\label{eq:C_ORA} C_{\text{ORA}} =
\text{log}_2(e)e^{\frac{1}{\gb}}
E_1\Bigl(\frac{1}{\gb}\Bigr),
\end{equation}
where $E_1(\cdot)$ is the exponential integral of first order~\cite[p.~xxxv]{gradshteyn/ryzhik00}.
Furthermore, if we include \emph{continuous} power adaptation, the
channel capacity, $C_{\text{OPRA}}$,
becomes~\cite{goldsmith/varaiya97,alouini/goldsmith99}
\begin{equation}
\label{eq: C_OPRA} C_{\text{OPRA}} = \text{log}_2(e)
\Bigl(\frac{e^{\frac{-\g_\text{cut}}{\gb}}}{\frac{\g_\text{cut}}{\gb}}-\gb\Bigr),
\end{equation}
where the ``cutoff'' value $\g_\text{cut}$ can be found by solving
\begin{equation}
\int_{\g_\text{cut}}^{\infty}\Bigl(\frac{1}{\g_\text{cut}}-\frac{1}{\gamma}\Bigr)f_{\g}(\g)
\, \dif \g=1.
\end{equation}
Thus, MASA$_N$ is compared to $C_{\text{ORA}}$, while
$\text{MASA}_{N\times K}$ and $\text{MASA}_{N\times \infty}$ are measured against
$C_{\text{OPRA}}$.
The capacity in \eqref{eq: C_OPRA} can be achieved in the case
that a continuum of capacity-achieving codes for AWGN channels, and
corresponding optimal power levels, are available. That is, for each
SNR there exists an optimal code and power level. Alternatively, if the fading is ergodic
within each codeword, as opposed to the assumptions in this paper, $C_{\text{OPRA}}$ can be obtained by a fixed rate transmission system using
a single Gaussian code~\cite{caire/shamai99, biglieri/proakis/shamai98}.

As the number of codes (switching thresholds) goes to infinity,
MASA$_N$ will reach the $C_\text{ORA}$ capacity,
while $\text{MASA}_{N\times K}$ will reach the $C_\text{OPRA}$
capacity when $N,K\to \infty$. Of course this is not a
practically feasible approach; however, as illustrated in
Figs.~\ref{fig: MASA_N C_ORA} and~\ref{fig: MASA 4 codes},
a small number of optimally designed codes, and possibly power adaptation levels, will indeed yield
a performance that is close to the theoretical upper bounds,  $C_\text{ORA}$ and $C_{\text{OPRA}}$,
for any given $\gb$.

From Fig.~\ref{fig: MASA 4 codes} we see
that the power adapted MASA schemes perform close to the theoretical
upper bound ($C_\text{OPRA}$) using only four codes. Specifically,
restricting our adaptive policy to just four rates and four power
levels per rate results in a spectral efficiency that is within $1$ dB
of the efficiency obtained with continuous-rate and
continuous-power~\eqref{eq: C_OPRA}, demonstrating the remarkable impact of
power adaptation.
This is in contrast to the case of continuous
rate adaptation, where introducing power adaptation gives negligible
gain~\cite{goldsmith/varaiya97}.

\subsection{Probability of no transmission}
\label{sec: outage-probability}
When the pre-adaptation SNR falls below $\swl{1,1}$ no data are sent. The probability of no transmission $\Pnotr$ for the Rayleigh fading case can then be calculated as follows,
\begin{equation}
  \Pnotr = \int_{0}^{\swl{1,1}}f_\g(\g)\, \dif\g
  = 1 - e^{-\frac{\swl{1,1}}{\gb}}
  \text{.}
  \label{eq:Po}
\end{equation}
When the number of codes is increased, the SNR range will be
partitioned into a larger number of regions.  As shown in
Fig.~\ref{fig: Switching Thresholds}, the
lowest switching level $\swl{1,1}$ will then become smaller.  $\Pnotr$ will
therefore decrease, as illustrated in Fig.~\ref{fig: P_out}. Similarly,
as seen from Fig.~\ref{fig: Switching Thresholds}, $\swl{1,1}$ also decreases with an increasing number of power levels, when
$N$ is constant. Thus, both rate and power adaptation flexibility reduces the probability of no transmission.

\begin{figure}
\centering
\includegraphics[width=3.2in]{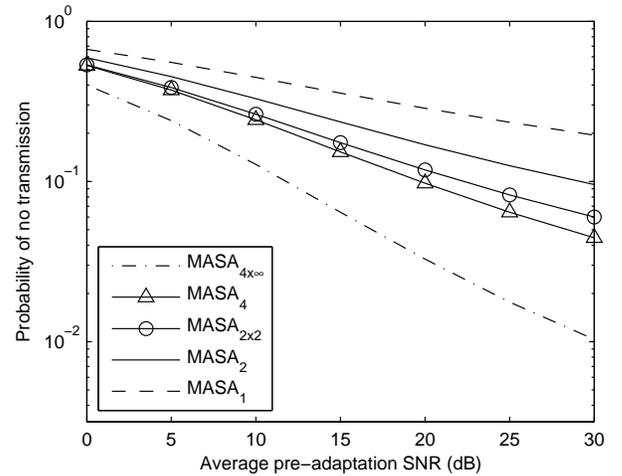}
\caption{The probability of no transmission $\Pnotr$ as a function of average pre-adaptation SNR.}
\label{fig: P_out}
\end{figure}


For applications with low delay requirements, it could be beneficial to enforce a constraint that $\Pnotr$ should not exceed
a prescribed maximal value. Then, we may simply---using~\eqref{eq:Po}---compute $\swl{1,1}$ to be the highest SNR
value which ensures that this constraint is fulfilled.
The MASA schemes are then optimized to obtain the highest possible ASE under the
given constraint on no transmission, i.e., optimization with $\swl{1,1}$ as a pre-determined parameter.
As an example, in Fig.~\ref{fig: MASA_Po}, the obtainable average spectral efficiency for the MASA$_N$ scheme with the
additional constraint that $\Pnotr \leq 10^{-3}$ (dashed lines) is compared to the
case without a constraint on no transmission probability (solid lines). We see that for $N=2$ the constraint
has a severe influence on the ASE, while for $N=8$ the constraint can be fulfilled without significant
losses in spectral efficiency.
\begin{figure}
  \centering
  \includegraphics{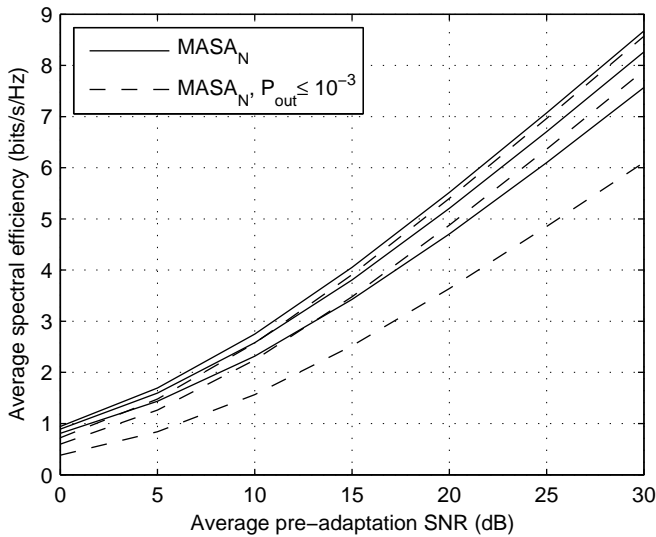}
  \caption{MASA$_N$ as a function of $\gb$, with a constraint on the probability of no transmission (solid lines) and without (dashed lines).
  Plotted for $N=2$ (lowermost curve for both series), $4$, and $8$ (uppermost curve for both series ).}
  \label{fig: MASA_Po}
\end{figure}


\section{Conclusions and Discussions}
\label{sec: conclusions}
Using a zero information outage approach, and assuming that capacity-achieving component codes are available, we have devised spectral efficiency maximizing link adaptation schemes
for flat block-fading wireless communication channels.
Constant, discrete, and continuous-power adaptation schemes are proposed
and analyzed. Switching levels and power adaptation policies are optimized in order to maximize the average spectral efficiency for a given fading distribution.

We have shown that a performance close to the Shannon limits can be achieved
with all schemes using only a small number of codes.
However, utilizing power adaptation is shown to give significant average spectral efficiency and probability of no transmission gains over the constant transmission power scheme.
In particular, using a fully discrete scheme with just four codes, each associated with four power levels,
we achieve a spectral efficiency within $1\,\dB$ of the Shannon
capacity for continuous rate and power adaptation. Additionally, constant and discrete-power
adaptation schemes render the system more robust against imperfect channel estimation and prediction, reduce
the feedback load and resolve implementation issues, compared to continuous power adaptation.

We have also seen that the number of rates $N$ can be traded against the number of power levels $K$.
This flexibility is of practical importance since it may be easier to implement the
proposed power adaptation schemes than to design capacity-achieving
codes for a large number of rates.
The analysis can be augmented to encompass more practical scenarios, e.g., by taking imperfect CSI~\cite{jetlund/oien/holm/hole04}
and SNR margins due to various implementation losses, into account. Finally, we note that the adaptive power algorithms presented in this paper
require that the radio frequency (RF) power amplifier is operated in the linear region, implying a higher power consumption. For devices with limited battery capacity it is apparent that there will be a tradeoff between efficiency and linearity. This can be a topic for further research.

\section*{Acknowledgement}
\label{sec:acknowledgement}
The authors wish to express their gratitude to Professor Tom Luo,
University of Minnesota, for suggesting the modified optimization
when the first switching level is constrained due to requirements on the probability of no transmission.
A similar idea has independently been proposed by Dr. Ola Jetlund, NTNU~\cite{jetlund/oien/holm/hole04}.



\bibliography{../../../referanser/referanser}


\end{document}